\documentclass{elsarticle}

\usepackage{amssymb,amsmath,amsthm,amsfonts}
\usepackage{color}

\newtheorem{theorem}{Theorem}[section]

\newtheorem{lemma}[theorem]{Lemma}
\newtheorem{proposition}[theorem]{Proposition}

\newcommand{\sinc}{\textnormal{sinc}}
\newcommand{\supp}{\textnormal{supp}\,}

\begin{document}

\begin{frontmatter}

\title{Average sampling of band-limited stochastic processes}

\author{Gilles Fa\"y}
\ead{gilles.fay@ecp.fr}

\author{Sinuk Kang\corref{cor}\fnref{tel}}
\ead{sinuk.kang@ecp.fr}

\address{Laboratory of Mathematics Applied to Systems, Ecole Centrale Paris, Grande voie de vignes, 92290 Cht\^{a}tenay Malabry, France}

\cortext[cor]{Corresponding author}
\fntext[tel]{Telephone: +33 1 41 13 17 89, Fax: +33 1 41 13 17 35}

\begin{abstract}
We consider the problem of reconstructing a wide sense stationary band-limited process from its local averages taken either at the Nyquist rate or above. As a result, we obtain a sufficient condition under which average sampling expansions hold in mean square and for almost all sample functions. 
Truncation and aliasing errors of the expansion are also discussed.
\end{abstract}

\begin{keyword}
average sampling \sep wide sense stationary stochastic process \sep sampling theorem
\MSC[2010] 42C15 \sep 94A12
\end{keyword}

\end{frontmatter}

\section{Introdution}

The classical Whittaker-Shannon-Kotel'nikov (WSK) sampling theorem \cite{Shannon:1949uo,Whittaker:1915wb,Kotelnikov:1933vx} says that a signal $f(t)$ in $PW_{\pi\omega}$, the Paley-Wiener space of functions band-limited to $[-\pi\omega,\pi\omega]$ (to be specified in Section \ref{notationdefinition}), is uniquely determined by its discrete samples $f({n}/{\omega})$'s, $n\in\mathbb Z$, and can be reconstructed via 
\begin{equation*}
	f(t) = \sum_{n \in \mathbb Z } f(\frac{n}{\omega}) \sinc (\omega t -n)
\end{equation*} which converges in $L^2(\mathbb R)$ and uniformly and absolutely on $\mathbb R$, where $\sinc t := {\sin \pi t}/({\pi t})$.

It is well known that the WSK sampling theorem has its counterpart for stochastic processes. The counterpart for wide sense stationary band-limited processes, definition of which is to be specified in Section \ref{notationdefinition}, is first introduced by Balakrishnan \cite{Balakrishnan:1957gx}, and developed further by many more authors, among them Lloyd \cite{Lloyd:1959vh} and Beutler \cite{Beutler:1961wk}.
Note that the sampling expansion in \cite{Lloyd:1959vh,Beutler:1961wk} is shown to converge both in mean square and with probability 1, while the expansion in \cite{Balakrishnan:1957gx} is shown to converge in mean square only. There also has been an effort to extend the class of stochastic processes to which the WSK sampling theorem applies. Zakai \cite{ZAKAI:1965ha} extended the notion of band-limited stochastic processes and proved the WSK sampling theorem hold over this extended notion. Using Zakai's technique, Lee \cite{Lee:1976uw} extended it to the class of second order measurable mean square continuous processes whose covariances are polynomially bounded.
Belyaev \cite{Belyaev:1959iv} and Piranashvili \cite{Piranashvili:1967bi} defined classes of the analytic processes, for which almost all trajectories can be analytically continued. Then they derived the WSK sampling expansion which holds either with probability 1 or for almost all sample functions in their classes.

Recently, an average sampling expansions (ASE's) for stochastic processes have been investigated in \cite{Song:2007fu,He:2011ep,Olenko:2011db}. Since acquisition devices do not produce signal values at the exact instances, in practice it seems more reasonable to use local averages instead of point evaluations. We call this sampling procedure the average sampling.
The ASE on band-limited functions was first presented by Gr\"{o}chenig \cite{Grochenig:1992ht}, and then {extended} in \cite{Aldroubi:2002cz,Feichtinger:1994td,Sun:2002eq}. 
In \cite{Grochenig:1992ht}, the author treated the problem of reconstructing a band-limited function from its local averages $\langle f , u_n \rangle = \int f(t)u_n(t) dt$ around $t_n$, where $\{u_n(t) : n \in \mathbb Z\}$ is a sequence of average functions satisfying 
\begin{equation}\label{v1.0condaverage}
	\supp  u_n(t) \subset [t_n - \frac{\delta}{2}, t_n + \frac{\delta}{2}],~ 0 \leq u_n(t),~\textnormal{and } \int_{-\infty}^{\infty} u_n (t)dt = 1.	
\end{equation} 
More precisely, he proved that for any $f(t)$ in $PW_{\pi\omega}$, if $0< t_{n+1} - t_n \leq \delta < {1}/({\sqrt{2}\pi\omega})$ then $f(t)$ is uniquely determined by the local averages $\langle f , u_n \rangle$'s, $n \in \mathbb Z$, and can be reconstructed by some iteration scheme. 
Applying the iteration scheme with $\{ \langle f , u_n \rangle: n \in \mathbb Z\}$ is equivalent to finding a frame expansion of $f$ on $PW_{\pi\omega}$, coefficients of which correspond to $\{ \langle f , u_n \rangle: n \in \mathbb Z\}$. That is, if the aforementioned condition is satisfied then there exists a frame $\{ r_n (t) : n \in \mathbb Z\}$ of $PW_{\pi\omega}$ such that
\begin{equation}\label{v1.0eq0.1}
	f(t) = \sum_{n \in \mathbb Z} \langle f , u_n \rangle r_n (t), ~ f \in PW_{\pi\omega}
\end{equation}
which converges in $L^2 (\mathbb R)$ and pointwise on $\mathbb R$. {Note here that $\{ \mathcal{P}u_n : n \in \mathbb Z\}$ and $\{r_n : n \in \mathbb Z\}$ are dual frame pairs of $PW_{\pi\omega}$ where $\mathcal{P}$ denotes the orthogonal projection of $L^2(\mathbb R)$ onto $PW_{\pi\omega}$.}

Using Gr\"{o}chenig's result, Song et al. \cite{Song:2007fu} addressed an ASE for band-limited stochastic processes. To be precise, let $X(t)$, $-\infty < t < \infty$, be a wide sense stationary stochastic process band-limited to $[-\pi\omega, \pi \omega]$ and $R_X(t)$ be its autocovariance function. Under the same notation and the assumption as in Gr\"{o}chenig's, they proved that if 
\begin{equation}\label{v2.3eq2}
	\{ t_n : n \in \mathbb Z\} \textnormal{ {is} relatively separable,}
\end{equation}	
 i.e., there is some constant $N$ such that $[k,k+1] \cap \{ t_n : n \in \mathbb Z\}$ contains at most $N$ elements for all $k \in \mathbb Z$, and  
\begin{equation} \label{v1.0eq1.2}
	|R_X (t)| \leq R_X(0) (1+|t|)^{-\eta}~ \textnormal{ for some } \eta >1,
\end{equation}
 then 
\begin{equation} \label{v2.1averageexp}
	X(t) = \sum_{n \in \mathbb Z} \langle X, u_n \rangle r_n (t)
\end{equation} which converges in mean square, 
i.e., 
\begin{equation*}
	\lim_{N \rightarrow \infty} E \Big| X(t) - \sum_{n =-N}^{N} \langle X,u_n\rangle r_n (t) \Big|^2 = 0,
\end{equation*}
for any $t \in \mathbb R$. It is assumed here that $\{r_n (t) : n \in \mathbb Z\}$ is a frame of $PW_{\pi\omega}$ for which \eqref{v1.0eq0.1} holds. 
Later, He et al. \cite{He:2011ep} provided an ASE to approximate wide sense stationary band-limited processes and, more generally, Olenko et al. \cite{Olenko:2011db} presented an ASE to approximate Piranashvili's processes. But their ASE's are asymptotically equivalent to the WSK sampling expansion, i.e., a point sampling.

It should be noted that the condition \eqref{v1.0eq1.2} seems too strong. For instance, consider class of linear combination{s} of $\{ \sinc (\omega t-n) : n \in \mathbb Z\}$, denoted by $PW_{\pi\omega}^o$. Then $PW_{\pi\omega}^o$ is dense in $PW_{\pi\omega}$, i.e., any function in $PW_{\pi\omega}$ can be approximated {in $L^2(\mathbb R)$} by a limit of a sequence of functions in $PW_{\pi\omega}^o$, and for any $f(t) \in PW_{\pi\omega}^o$, $|f(t)|(1+|t|)^{\eta} \rightarrow \infty$ as $|t| \rightarrow \infty$ whenever $\eta >1$. {In fact, one can prove Theorem 2.2 of \cite{Song:2007fu} without assuming the condition \eqref{v2.3eq2} and \eqref{v1.0eq1.2}, based on the observation that \eqref{v1.0eq0.1} converges unconditionally on $\mathbb R$ (see Theorem \ref{v1.0thm2.2}).}

In this paper, we show that the average sampling theorem given in \cite{Song:2007fu} remains true without the condition \eqref{v2.3eq2} and \eqref{v1.0eq1.2}.
We also provide an average sampling theorem with local averages taken at the Nyquist rate, while only oversampled local averages were considered in \cite{Song:2007fu}. Success of a perfect reconstruction from local averages via the resulting expansion depends only on the length of support of average functions $u_n$, which improves the results of \cite{Song:2007fu,He:2011ep}. 
The latter ASE of ours not only converges in mean square but also converges for almost all sample functions. 
Under band-guard condition (see e.g. \cite{Helms:1962wo}) we derive explicit upper bounds on the truncation error of the ASE. Aliasing error is also discussed.

This paper is organized as follows. In Section \ref{notationdefinition} we introduce notations and definitions needed throughout the paper. In Section \ref{secASO} we present an average sampling theorem in which local averages are taken above the Nyquist rate (oversampling). In Section \ref{Nyquist} we show that band-limited stochastic processes can also be reconstructed by its local averages taken at the Nyquist rate. It is shown that the resulting ASE converges both in mean square and for almost all sample functions. Finally in Section \ref{error} truncation and aliasing errors of the expansion are discussed.

\section{Notations and definitions}\label{notationdefinition}

The Paley-Wiener space of signals band-limited to $[-\pi\omega, \pi\omega]$ is defined by
\begin{equation*}
	PW_{\pi\omega} := \{ f \in L^2(\mathbb R) \cap C(\mathbb R) : \textnormal{supp} \hat{f} \subset [-\pi\omega, \pi\omega] \}
\end{equation*}
where we define the Fourier transform as $\mathcal{F}[f](\xi) = \hat{f}(\xi) := \int_{-\infty}^{\infty} f(t) e^{-it\xi} dt$, $f \in L^1(\mathbb R)$, and extend it to an isomorphism from $L^2(\mathbb R)$ to $L^2(\mathbb R)$.

A sequence $\{\phi_n:n\in \mathbb{Z}\}$ of vectors in a separable Hilbert space $\mathcal{H}$ equipped with the norm $\| \cdot \|_\mathcal{H}$ is 
\begin{itemize}
\item a frame of $\mathcal{H}$ with bounds $(A,B)$ if there are constants $B \geq A >0$ such that
$$
A\Vert f \Vert^{2}_\mathcal{H}\leq \sum_{n\in \mathbb{Z}}|\langle f ,\phi _{n}\rangle_\mathcal{H}|^{2}\leq B\Vert f\Vert^{2}_\mathcal{H},~f \in \mathcal{H};
$$
\item a Riesz (or stable) basis of $\mathcal{H}$ with bounds $(A,B)$ if $\{\phi_n:n\in \mathbb{Z}\}$ is complete in $\mathcal H$ and there are constants $B \geq A >0$ such that
$$
A\|\mathbf{c}\|^{2}\leq  \Big\| \sum_{n\in \mathbb{Z}}c(n)\phi _{n}\Big\|_\mathcal{H}^{2}\leq B\| \mathbf{c} \|^{2},~\mathbf c := \{ c(n)\}_n \in \ell^2 (\mathbb Z)
$$ where $\|\mathbf{c}\|^{2} = \sum_{n \in \mathbb Z} |c(n)|^2$.
\end{itemize}

A stochastic process $\{ X(t) : t \in \mathbb R \}$ is wide sense stationary if $E(X(t)) = 0$ and $E|X(t)|^2 < \infty$ for $t \in \mathbb R$ and the autocovariance function $R_X (t,s) :=E(X(t) \overline{X(s)})$ depends only on the difference $t-s$. By the spectral representation theorem \cite{Doob:1990us}, a wide sense stationary process $\{ X(t) : t \in \mathbb R \}$ has a spectral representation: $X(t) = \int_{-\infty}^{\infty} e^{it\lambda} dy(\lambda)$, $ t \in \mathbb R$, where the process $y$ has orthogonal increments and $F$ is the spectral distribution function of $\{ X(t) : t \in \mathbb R \}$ such that $E|dy(\lambda)|^2 = dF(\lambda)$. 

$\{ X(t) : t \in \mathbb R \}$ is said to be band-limited to $[-\pi\omega, \pi\omega]$ if its spectrum {has the support} $[-\pi\omega, \pi\omega]$, i.e., $X(t) = \int_{-\pi\omega}^{\pi\omega}e^{it\lambda} dy(\lambda)$, $t \in \mathbb R$. {Since $X(t) = \int_{-\pi\omega}^{\pi\omega}e^{it\lambda} dy(\lambda)$ if and only if $R_X (t,s) =  \int_{-\pi\omega}^{\pi\omega}e^{i(t-s)\lambda} dF(\lambda) $ \cite{Cramer:1940vr}, $\{ X(t) : t \in \mathbb R \}$ is also said to be band-limited to $[-\pi\omega,\pi\omega]$ if $R_X(t)$ is band-limited to $[-\pi\omega,\pi\omega]$, i.e., $R_X (t) \in PW_{\pi\omega}$.}

We denote by $\langle \cdot, \cdot \rangle$ the usual inner product in $L^2(\mathbb R)$ unless otherwise specified.

\section{Oversampled local averages} \label{secASO}

We extend Theorem 2.2 of \cite{Song:2007fu} into the following theorem by removing the aforementioned condition \eqref{v1.0eq1.2} and the constraint for $\{ t_n : n \in \mathbb Z\}$ being relatively separable.
\begin{theorem} \label{v1.0thm2.2}
Let $\{u_n (t) : n \in \mathbb R\}$ be a sequence of average functions satisfying \eqref{v1.0condaverage}.
For a wide sense stationary process $X(t)$ band-limited to $[-\pi\omega,\pi\omega]$, if $t_{n+1} - t_n \leq \delta < {1}/({\sqrt{2}\pi\omega})$ then
\begin{equation} \label{v1.0thm2.1eq4}
		X(t)  =   \sum_{n \in \mathbb Z} \langle X,u_n \rangle r_n (t)
\end{equation}
which converges in mean square, 
{uniformly on any compact subset of $\mathbb R$,} where $\{ r_n (t): n \in \mathbb Z\}$ is a frame of $PW_{\pi\omega}$ for which \eqref{v1.0eq0.1} holds.
\end{theorem}

\begin{proof}
Since $E|X(t)|^2 <\infty$ for $t \in \mathbb R$, we have by Fubini's theorem that for any $t \in \mathbb R$ and any positive integer $N$, 
	\begin{eqnarray} \label{v1.0eq2}
		&& E \Big| X(t) - \sum_{n =-N}^{N} \langle X,u_n\rangle r_n (t) \Big|^2 \nonumber \\ 
		&& = R_X (0) -  \sum_{n =-N}^{N} \langle R_X(t-\cdot) ,u_n(\cdot) \rangle r_n (t) 
		  -  \sum_{n =-N}^{N} \overline{ \langle R_X(t-\cdot) ,u_n(\cdot) \rangle r_n (t)} \nonumber \\
		&& +  \sum_{k,n =-N}^{N}  \int_{t_k-\frac{\delta}{2}}^{t_k+\frac{\delta}{2}} \int_{t_n-\frac{\delta}{2}}^{t_n+\frac{\delta}{2}} {R_X(x-y)}u_n (x) u_k (y) dx dy \,r_k(t) \overline{r_n (t)}.
	\end{eqnarray}
We first show that the double summation \eqref{v1.0eq2} converges unconditionally for any $t \in \mathbb R$ as $N$ goes to infinity.
Since $R_X(t'-\cdot) \in PW_{\pi\omega}$ for $t' \in \mathbb R$, we obtain from \eqref{v1.0eq0.1} that {for a given $t' \in \mathbb R$}
\begin{equation} \label{v1.0eq5.1}
	R_X(t'-t) = \sum_{k \in \mathbb Z} \int_{t_k-\frac{\delta}{2}}^{t_k+\frac{\delta}{2}} R_X(t'-y)u_k(y)dy \,r_k (t).
\end{equation} 
It should be noticed that \eqref{v1.0eq5.1} converges unconditionally {for each $t\in\mathbb R$ and uniformly on $\mathbb R$} since \eqref{v1.0eq0.1} is a frame expansion in $PW_{\pi\omega}$ so that it converges unconditionally both in $L^2 (\mathbb R)$ and pointwise on $\mathbb R$ {and moreover the pointwise convergence is uniform}{: this follows by Corollary 3.1.5 of \cite{Christensen:2008wy} together with $PW_{\pi\omega}$ being a reproducing kernel Hilbert space {with bounded reproducing kernel} \cite{Higgins:1996uoa}.}
Setting $t = t'$, we have 
\begin{equation}\label{v3.1sec3eq1}
	R_X(0) = \sum_{k \in \mathbb Z} \int_{t_k-\frac{\delta}{2}}^{t_k+\frac{\delta}{2}} R_X(t-y)u_k(y)dy \,r_k (t)
\end{equation}
which converges unconditionally for any $t\in \mathbb R$. {Furthermore, \eqref{v3.1sec3eq1} converges uniformly on any compact subset of $\mathbb R$ (Theorem 7.13 of \cite{Rudin:1964wf}).}

Now, for a given $x \in \mathbb R$, let $a_{k,x}(t) :=  \int_{t_k-{\delta}/{2}}^{t_k+{\delta}/{2}}{R_X(x-y)} u_k (y) dy \,r_k (t)$, $t \in \mathbb R$. Then $\sum_{k \in \mathbb Z} a_{k,x}(t)$ converges unconditionally to $R_X (x-t)$ for any $t \in \mathbb R$ and uniformly on $\mathbb R$ (with respect to t). Thus
\begin{eqnarray*}
	\lim_{N \rightarrow \infty} \sum_{n=-N}^{N} \int_{t_n-\frac{\delta}{2}}^{t_n+\frac{\delta}{2}} \Big( \sum_{k \in \mathbb Z} a_{k,x}(t) \Big)u_n (x) dx\, \overline{r_n (t)} 
	=  \sum_{n\in\mathbb Z} \int_{t_n-\frac{\delta}{2}}^{t_n+\frac{\delta}{2}}  R_X(x-t) u_n (x) dx\, \overline{r_n (t)}
\end{eqnarray*}
converges unconditionally to $\overline{R_X(0)}$ for any $t \in \mathbb R$ and uniformly on any compact subset of $\mathbb R$, so that the double summation \eqref{v1.0eq2} converges unconditionally for any $t \in \mathbb R$ and uniformly on any compact subset of $\mathbb R$ as $N$ goes to infinity. 

{Since the absolute convergence is equivalent to the unconditional convergence for real- or complex-valued series (see e.g. Lemma 3.3 of \cite{Heil:2011jr}),} 
\begin{equation} \label{v1.7theproof}
	\lim_{N \rightarrow \infty} \sum_{k,n =-N}^{N}  \int_{t_k-\frac{\delta}{2}}^{t_k+\frac{\delta}{2}} \int_{t_n-\frac{\delta}{2}}^{t_n+\frac{\delta}{2}} {R_X(x-y)}u_n (x) u_k (y) dx dy\, {r}_k(t) \overline{r_n (t)}
\end{equation} 
converges absolutely for any $t \in \mathbb R$.

Thus it follows by \eqref{v1.0eq5.1} that for any $t \in \mathbb R$
\begin{eqnarray*}
	\eqref{v1.7theproof} 
	& = &  \sum_{n \in \mathbb Z}  \int_{t_n-\frac{\delta}{2}}^{t_n+\frac{\delta}{2}} \Big( \sum_{k \in \mathbb Z} \int_{t_k-\frac{\delta}{2}}^{t_k+\frac{\delta}{2}} {R_X(x-y)} u_k (y)  dy \,{r}_k(t)\Big) \, u_n (x) ds\,  \overline{r_n (t)} \\
	& = &   \sum_{n \in \mathbb Z}  \int_{t_n-\frac{\delta}{2}}^{t_n+\frac{\delta}{2}}R_X (x-t) \, u_n (x) ds \, \overline{r_n (t)} \\
	& = & \overline{R_X(0)}
\end{eqnarray*}
from which we have $\lim_{N \rightarrow \infty} E \Big| X(t) - \sum_{n =-N}^{N} \langle X,u_n\rangle r_n (t) \Big|^2 = 0$, {uniformly on any compact subset of $ \mathbb R$.}
\end{proof}

The globally uniform convergence of \eqref{v1.0thm2.1eq4} is not guaranteed in general. It is worth mentioning that the necessary and sufficient condition for the WSK sampling expansion of so-called I-process, i.e., a band-limited stochastic process possessing an absolutely continuous spectral distribution function, to be globally uniformly convergent is addressed in \cite{Boche:2010dv}.

The WSK sampling theorem for band-limited stochastic processes \cite{Balakrishnan:1957gx,Lloyd:1959vh,ZAKAI:1965ha} states that any $X(t)$ band-limited to $[-\pi\omega, \pi\omega]$ can be reconstructed by 
\begin{equation*}
	X(t) = \sum_{n \in \mathbb Z} X(\frac{n}{\omega})\sinc(\omega t - n) 
\end{equation*} which converges in mean square or with probability 1 for any $t \in \mathbb R$. In this case, the samples $\{X({n}/{\omega}):n \in \mathbb Z\}$ are taken at the Nyquist rate $\omega$. 
However, Theorem \ref{v1.0thm2.2} does not cover the case of local averages taken at the Nyquist rate, $\{ \langle X, u_n \rangle : n \in \mathbb Z \}$, since if $t_n = {n}/{\omega}$ then $t_{n+1} - t_n = {1}/{\omega} > {1}/({\sqrt{2}\omega\pi})$, $n \in \mathbb Z$, where $\supp u_n \subset [t_n -{\delta}/{2}, t_n + {\delta}/{2}]$.

\section{Local averages taken at the Nyquist rate}\label{Nyquist}
In this section we consider the case of local averages $\{ \langle X, v_n \rangle : n \in \mathbb Z \}$ taken at the Nyquist rate. 
In what follows, we assume $\omega = 1$. The aim of Section \ref{Nyquist} is to derive an ASE of the form:
\begin{equation}\label{v3.3sec3eq1}
	X(t) = \sum_{n \in \mathbb Z} \langle X,v_n \rangle s_n(t)
\end{equation} which converges in a proper sense. Here, $\{ s_n (t) : n \in \mathbb Z \} $ is a frame of $PW_\pi$ and $\{ v_n(t) : n \in \mathbb Z\}$ is a sequence of average functions satisfying
\begin{eqnarray} \label{v1.0eq7.1}
	&& \textnormal{supp}\, v_n \subset [n-a,n+b] ~\textnormal{ for } a,b \geq 0 \textnormal{ and } a+b > 0, \\ 
	&& 0 \leq v_n  \in L^2 (\mathbb R), \textnormal{ and } \int_{-\infty}^{\infty} v_n (t) dt = \int_{n-a}^{n+b} v_n (t) dt = 1,~n \in \mathbb Z. \nonumber
\end{eqnarray}

Our main results, Theorem \ref{v1.0thm2.4} and \ref{v1.0thm2.5}, are based on the following ASE in $PW_\pi$. 
\begin{proposition}[Theorem 3.2. of \cite{Kang:2011bf}] \label{v1.0prop2.3}
Let $\{v_n (t) : n \in \mathbb R\}$ be a sequence of average functions satisfying \eqref{v1.0eq7.1}
and let $\delta := \max\{a,b\}$. 
If $\sqrt{\delta(a+b)} < {1}/{\pi}$, then there is a frame $\{ s_n (t) : n \in \mathbb Z\} $ of $PW_{\pi}$ such that 
\begin{equation} \label{v1.4propeq8}
	f(t)= \sum_{n \in \mathbb Z} \langle f, v_n \rangle s_n (t), ~f \in PW_{\pi}
\end{equation}
which converges in $L^2(\mathbb R)$ and uniformly and absolutely on $\mathbb R$. {In this case,  $\{ s_n (t) : n \in \mathbb Z\} $ and $\{ \mathcal{P} v_n (t) : n \in \mathbb Z \} $ are dual frame pairs of $PW_\pi$ where $\mathcal{P}$ is the orthogonal projection of $L^2(\mathbb R)$ onto $PW_\pi$. }
\end{proposition}
\begin{proof}
Let $\phi(t) := \sinc(t)$. Note that $\phi$ is differentiable, $\phi' \in L^2 (\mathbb R)$, $|Z_\phi (0,\xi)| = 1$ for $\xi \in \mathbb R$, and $\| Z_{\phi'}(t,\xi) \|_{L^\infty(\mathbb R^2)} = \pi$, where $Z_f (t,\xi):= \sum_{n \in \mathbb Z} f (t-n) e^{in\xi}$ denotes the Zak transform \cite{JANSSEN:1988ul} of $f(t) \in L^2 (\mathbb R)$. Then Proposition \ref{v1.0prop2.3} is an immediate consequence of Theorem 3.2 of \cite{Kang:2011bf} with $\phi(t) = \sinc(t)$.
\end{proof}

Since $R_X(t'-\cdot) \in PW_\pi$ for a given $t' \in \mathbb R$, assuming $\sqrt{\delta(a+b)} < {1}/{\pi}$, we have by Proposition \ref{v1.0prop2.3} 
		\begin{equation} \label{v1.0sec4lem3.2}
			R_X (t'-t)  = \sum_{n \in \mathbb Z } \langle R_X (t'-\cdot), v_n (\cdot) \rangle s_n (t),~ t \in \mathbb R 
		\end{equation}
which converges in $L^2 (\mathbb R)$ and absolutely and uniformly on $\mathbb R$. As already mentioned in the proof of Theorem \ref{v1.0thm2.2}, the absolute convergence is equivalent to the unconditional convergence for real- or complex-valued series, so \eqref{v1.0sec4lem3.2} also converges unconditionally on $\mathbb R$.

With the same definition as in Proposition \ref{v1.0prop2.3}, we have a counterpart statement for stochastic processes:
\begin{theorem}\label{v1.0thm2.4}
Let $\{ X(t) : t \in \mathbb R \}$ be a wide sense stationary process band-limited to $[-\pi, \pi]$. If $\sqrt{\delta(a+b)} < {1}/{\pi}$, then \eqref{v3.3sec3eq1}
converges in mean square, {uniformly on any compact subset of $\mathbb R$,} where $\{ s_n (t) : n \in \mathbb Z \}$ is a frame of $PW_{\pi}$ for which \eqref{v1.4propeq8} holds.
\end{theorem}
\begin{proof}
The proof is essentially the same as the proof of Theorem \ref{v1.0thm2.2}.
\end{proof}

Furthermore, we prove the ASE \eqref{v3.3sec3eq1} converges for almost all sample functions.
\begin{theorem}\label{v1.0thm2.5}
Let the notation and the assumption be the same as in Theorem \ref{v1.0thm2.4}. Then \eqref{v3.3sec3eq1} holds for almost all sample functions.
\end{theorem}

\begin{proof}
Let 
\begin{equation}\label{v1.9thmproofeq}
	\tilde{X}(t)  := \lim_{N \to \infty} \sum_{|n| < N } \langle X,v_n  \rangle s_n (t)
\end{equation} in mean square sense.
We have shown in the proof of Theorem \ref{v1.0thm2.2} (with a proper modification of the notation) that for any $t \in \mathbb R$
\begin{eqnarray*} \label{v1.0eqepsilon}
	E|X(t) - \tilde{X}(t) |^2  =   \lim_{N \rightarrow \infty} E|X(t) - \sum_{n =-N}^N \langle X,v_n  \rangle s_n (t) |^2.
\end{eqnarray*}
Thus, $X(t) = \tilde{X}(t)$ with probability 1 by Theorem \ref{v1.0thm2.4}. The right of \eqref{v1.9thmproofeq} converges uniformly {on any compact subset of $\mathbb R$} and $X(t)$ is continuous for almost all sample functions:  in fact, almost all sample functions are entire functions \cite{Belyaev:1959iv}. Then the theorem follows.
\end{proof}

\section{Error estimation}\label{error}

The aim of this section is to estimate truncation and aliasing errors of the ASE \eqref{v3.3sec3eq1}.
To this end we always assume $\sqrt{\delta(a+b)} < {1}/{\pi}$ so that, by Theorem \ref{v1.0thm2.5}, 
$\{ \mathcal{P}v_n (t) : n \in \mathbb Z\}$ and $\{ s_n (t): n \in \mathbb Z\}$ are dual frame pairs of $PW_\pi$ for which \eqref{v3.3sec3eq1}
holds for almost all sample functions of a given wide sense stationary process $X(t)$ band-limited to $[-\pi,\pi]$. Here, $\mathcal{P}$ is the orthogonal projection of $L^2(\mathbb R)$ onto $PW_\pi$. Note that $\mathcal{P}$ is a shift-invariant operator, i.e., $\mathcal{P}[f(\cdot-n)](t) = \mathcal{P}f (t-n),~n \in \mathbb Z$, for any $f \in L^2(\mathbb R)$.

In the following we assume further that $u_n (t) = u(t-n)$ for $n \in \mathbb Z$. 
Then $\{ \mathcal{P}u_n(t) = \mathcal{P}u(t-n) : n \in \mathbb Z\}$ is a frame of $PW_\pi$ if and only if there exist constants $B \geq A >0$ such that 
\begin{equation}\label{v2.3sec4eq0.1}
	A \leq | \hat{u}(\xi) | \leq B \textnormal{ a.e. on } [-\pi,\pi]
\end{equation}
(Theorem 2 of \cite{Garcia:2005fz}). As a matter of fact, the condition \eqref{v2.3sec4eq0.1} is also a sufficient and necessary condition for $\{ \mathcal{P}u(t-n) : n \in \mathbb Z\}$ to be a Riesz basis of $PW_\pi$.

Since $\mathcal{P}u(t) \in PW_\pi$, we have by \eqref{v1.4propeq8}
\begin{equation*}
	\mathcal{P}u(t) = \sum_{n \in \mathbb Z} \langle \mathcal{P}u(\cdot) , u(\cdot -n) \rangle s(t-n)
\end{equation*}
so that, via the Fourier transform,
\begin{eqnarray*}
	\hat{u}(\xi)\chi_{[-\pi,\pi]}(\xi) 
	& = & \Big( \sum_{n \in \mathbb Z} \langle \hat{u}(\xi)\overline{\hat{u}(\xi)} , e^{-in\xi} \rangle_{L^2[-\pi,\pi]} e^{-in\xi} \Big) \hat{s}(\xi) \\
	& = & {| \tilde{\hat{u}}(\xi)|^2 \hat{s}(\xi)}
\end{eqnarray*}
which holds in $L^2(\mathbb R)$ {where $\tilde{\hat{u}}(\xi)$ is $2\pi$-periodic extension of $\hat{u}(\xi)\big|_{[-\pi,\pi]}$, the restriction of $\hat{u}(\xi)$ on $[-\pi,\pi]$.}
Thus we have
\begin{equation} \label{v2.3sec4eq1}
	\overline{\hat{s}(\xi)} = \frac{1}{{\hat{u}(\xi)}} \chi_{[-\pi,\pi]}(\xi) \textnormal{ a.e. on }  \mathbb R.
\end{equation}

In summary we consider the ASE of the form
 \begin{equation}\label{v2.6sec4eq2}
	{X}(t)  = \sum_{n \in \mathbb Z} \langle X(\cdot),u(\cdot-n)  \rangle s(t-n)
\end{equation}
where $u(t)$ and $s(t)$ satisfy \eqref{v2.3sec4eq0.1} and \eqref{v2.3sec4eq1}, respectively.

It is, by definition, unavoidable that $\hat{s}(\xi)$ has discontinuities at $\pm \pi$. Thus, for a given $t \in \mathbb R$, $s(t-n)$ decays slowly as $n$ goes to infinity so that convergence speed of  \eqref{v2.6sec4eq2} is also slow. To overcome this, we adapt so-called oversampling technique by the guard-band assumption. This is introduced in \cite{Yao:1966hf,Helms:1962wo,BrownJr:1969wr} to estimate truncation error bound of the WSK sampling expansion of band-limited functions. The same method is also applied to truncation error estimation of the WSK sampling expansion of band-limited stochastic processes \cite{BrownJr:1968uh}.

Consider $f(t) \in PW_\omega$ where $0< \omega <\pi$, i.e., $\supp \hat{f} \subset [-\omega,\omega]\subsetneq [-\pi,\pi]$. Since $PW_\omega \subseteq PW_\pi$, we have by Proposition \ref{v1.0prop2.3}
\begin{equation*}
	f(t) = \sum_{n \in \mathbb Z} \langle f( \cdot), u(\cdot-n) \rangle s(t-n)
\end{equation*}
which is, via the Fourier transform, equivalent to 
\begin{equation}\label{v2.8sec4eq1}
	\hat{f}(\xi) = \hat{f}(\xi)\theta(\xi) = \sum_{n \in \mathbb Z} \langle f( \cdot), u(\cdot-n) \rangle \hat{s}(\xi) \theta(\xi)e^{-in\xi}
\end{equation}
where $\theta(\xi)$ is an arbitrary smooth function satisfying $\theta(\xi) = 1$ on $[-\omega,\omega]$ and $\theta(\xi)=0$ on $\mathbb R \backslash [-\pi,\pi]$.
Applying the inverse Fourier transform on \eqref{v2.8sec4eq1} gives
\begin{equation}\label{v2.8sec4eq2}
	f(t) = \sum_{n \in \mathbb Z} \langle f( \cdot), u(\cdot-n) \rangle \tilde{s}(t-n)
\end{equation}
where 
\begin{equation}\label{v2.8sec4eq3}
	\hat{\tilde s} (\xi) = \hat s (\xi)\theta(\xi) = \frac{1}{\overline{\hat u (\xi)}} \theta (\xi) 
\end{equation}
in $L^2(\mathbb R)$. Note that $\tilde s (t) \in PW_\pi$.

\begin{lemma} \label{v2.3sec4lemma1}
For $\tilde s(t)\in PW_\pi$ satisfying \eqref{v2.8sec4eq3}, if $\theta(\xi)$ is $p$-times continuously differentiable for some integer $p>1$, i.e., $\theta(\xi) \in C^p (\mathbb R)$, then for any given $t \in \mathbb R$
\begin{equation}\label{v2.4sec4lemeq1}
	|\tilde s(t-n)| \leq \frac{C_p(t)}{|n|^p},~ n \in \mathbb Z 
\end{equation}
where 
\begin{equation} \label{v3.4sec4eq2}
	C_p(t) := \frac{1}{2\pi} \int_{-\pi}^{\pi} \Big| (\frac{1}{\hat{u}(\xi)} \theta(\xi)e^{-it\xi} )^{(p)} \Big| d\xi < \infty .
\end{equation}
\end{lemma}

\begin{proof}
Note first that $u(t)$ is compactly supported so that $\hat{u}(\xi)$ is infinitely many differentiable on $\mathbb R$.
Since 
\begin{equation*}
	\tilde s(t-n) =  \frac{1}{2\pi} \int_{-\pi}^{\pi}\hat{s}(\xi)\theta(\xi) e^{it\xi}  e^{-in\xi} d\xi = \frac{1}{2\pi} \int_{-\pi}^{\pi} \frac{1}{\overline{\hat{u}(\xi)}} \theta(\xi) e^{it\xi}  e^{-in\xi} d\xi
\end{equation*}
is the $n$-th coefficient of the Fourier series of $p$-times continuously differentiable function ${ \theta(\xi) e^{it\xi} }/{\overline{\hat{u}(\xi)}}$ vanishing at $\pm \pi$, on $[-\pi,\pi]$, we have
\begin{eqnarray*}
	\int_{-\pi}^{\pi}  (\frac{1}{\overline{\hat{u}(\xi)}}\theta(\xi)e^{it\xi} )^{(p)} e^{-in\xi} d\xi 
	= - (-in)^p \int_{-\pi}^{\pi}\frac{1}{\overline{\hat{u}(\xi)}}\theta(\xi)e^{it\xi} e^{-in\xi} d\xi 
	=  - 2\pi(-in)^p \tilde s(t-n)
\end{eqnarray*}
which implies \eqref{v2.4sec4lemeq1}.
\end{proof}

\begin{theorem}
Let $X(t)$ be a wide sense stationary process band-limited to $[-\omega, \omega] \subsetneq [-\pi,\pi]$ and $R_X(t)$ the autocovariance function of $X(t)$. Assume that ASE \eqref{v2.6sec4eq2} holds.
Then we have for any $t \in \mathbb R$ 
\begin{equation}\label{v2.8sec4eq5}
	X(t) = \sum_{n \in \mathbb Z} \langle X( \cdot), u(\cdot-n) \rangle \tilde{s}(t-n)
\end{equation}
which converges in mean square where $\tilde s(t)$ is given by \eqref{v2.8sec4eq3}.
Morever, if $\theta(\xi) $ in \eqref{v2.8sec4eq3} belongs to $ C^p (\mathbb R)$ for some integer $p>1$ then
\begin{equation*}
	E|X(t) - X_N (t) |^2 \leq \frac{4 R_X(0) C_p(t)^2 }{(p-1)^2 N^{2(p-1)}}
\end{equation*}
where 
\begin{equation*}
	X_N(t):= \sum_{|n| \leq N} \langle X(\cdot), u(\cdot -n) \rangle \tilde s(t-n)
\end{equation*}
and $C_p (t)$ is given by \eqref{v3.4sec4eq2}.
\end{theorem}

\begin{proof}
Since $R_X(t'-\cdot) \in PW_\omega$ for a given $t' \in \mathbb R$, it follows by \eqref{v2.8sec4eq2} that
	\begin{equation} \label{v2.8sec4eq4}
		R_X (t'-t)  = \sum_{n \in \mathbb Z } \langle R_X (t'-\cdot), u (\cdot-n) \rangle \tilde s(t-n), ~t \in \mathbb R
	\end{equation}
which converges in $L^2 (\mathbb R)$ and absolutely and uniformly on $\mathbb R$. Using \eqref{v2.8sec4eq4} we can obtain \eqref{v2.8sec4eq5} by the same argument as in the proof of Theorem \ref{v1.0thm2.2}. 

For any $t \in \mathbb R$ it follows that
\begin{eqnarray*}
	E|X(t) - X_N (t) |^2 
	&=& E\Big| \sum_{|n|>N} \langle X, u_n \rangle \tilde s_n (t)  \Big|^2 \\
	&=& \sum_{|n|>N} \int_{n-a}^{n+b} \Big( \sum_{|k|>N} \int_{k-a}^{k+b} R_X(x-y)u_k(y)dy \tilde s_k (t)  \Big) u_n (x)dx \, \overline{\tilde s_n (t)} \\
	& \leq &  \sum_{|n|>N} \int_{n-a}^{n+b} \left( \sum_{|k|>N} \Big| \int_{k-a}^{k+b} R_X(x-y)u_k(y)dy\Big| \Big| \tilde s_k (t)\Big|  \right) u_n (x)dx \, \overline{\tilde s_n (t)}\\
	& \leq & \sup_{-(a+b) \leq t \leq a+b} |R_X (t)|   \sum_{|n|>N} \int_{n-a}^{n+b} \Big( \sum_{|k|>N} | \tilde s_k (t)|  \Big) u_n (x)dx \, \overline{\tilde s_n (t)} \\
	& \leq &R_X (0)   \Big( \sum_{|k|>N} | \tilde s_k (t)|  \Big)^2
\end{eqnarray*}
in which $\tilde s_k(t) = \tilde s(t-k)$ and $u_k (t) = u(t-k)$ for $k \in \mathbb Z$.
By Lemma \ref{v2.3sec4lemma1} combined with the integral test, 
\begin{eqnarray*}
	 \sum_{|k|>N} | \tilde s(t-k)|  \leq 
	  \sum_{ k > N} \frac{2 C_p(t)}{k^p}
	 \leq   \frac{2 C_p(t)}{(p-1)N^{p-1}}
\end{eqnarray*}
which proves the theorem.	
\end{proof}

As already mentioned, $\{\mathcal{P}u(t-n) : n \in \mathbb Z\}$ and $\{ s(t-n) : n \in \mathbb Z\}$ are dual Riesz basis pairs of $PW_\pi$. Thus the orthogonal projection $\mathcal{P}$ of $L^2(\mathbb R)$ onto $PW_\pi$ can be written as 
\begin{equation*}
	\mathcal{P}f (t) = \sum_{n \in \mathbb Z} \langle f(\cdot), \mathcal{P}u(\cdot-n) \rangle s(t-n).
\end{equation*}
To consider an aliasing error of ASE \eqref{v3.3sec3eq1}, we extend $\mathcal{P}$ for stochastic processes as
\begin{equation*} \label{v2.6sec4proj}
	{\mathcal{P}}X(t) := \sum_{n \in \mathbb Z} \langle X(\cdot), \mathcal{P}u(\cdot-n) \rangle s(t-n).
\end{equation*}
Assuming that ASE \eqref{v3.3sec3eq1} holds for wide sense stationary processes band-limited to $[-\pi,\pi]$, one can easily see that 
for any $t \in \mathbb R$ 
\begin{equation*}
	E|X(t) - {\mathcal{P}}X(t) |^2 = \int_{|\lambda| > \pi} F(d\lambda)
\end{equation*}
where $F$ is the spectral distribution function of $X(t)$.

\section*{Acknowledgement}
Sinuk Kang is partially supported by Erasmus Mundus BEAM program funded by European Commission.

%

\end{document}